%% file: main.tex
\documentclass[%
reprint,
superscriptaddress,
prd,
 amsmath,amssymb,
 aps,
nofootinbib
]{revtex4-2}

\usepackage{graphicx}
\usepackage{svg}
\usepackage{amsthm}
\usepackage{xcolor}
\usepackage{physics}
\usepackage{siunitx}
\usepackage{caption}
\usepackage{subcaption}
\usepackage{hyperref}

\usepackage{cleveref}
\crefname{equation}{Eq.}{Eqs.}
\crefname{figure}{Fig.}{Figs.}
\crefname{table}{Tab.}{Tabs.}
\crefname{section}{Sec.}{Secs.}
\crefname{appendix}{App.}{Apps.}
\crefname{proposition}{Prop.}{Props.}
\Crefname{table}{Table}{Tables}
\Crefname{figure}{Figure}{Figures}

\definecolor{red}{rgb}{1.0, 0, 0}

\newcommand{\R}{\mathbb{R}}
\newcommand{\Z}{\mathbb{Z}}
\newcommand{\obs}{\mathcal{O}}

\DeclareMathOperator{\Var}{Var}
\newcommand{\pdf}{\text{PDF}}

\newtheorem{proposition}{Proposition}[section]

\usepackage{standalone}
\usepackage{tikz}
\usepackage{tikz-3dplot}
\usetikzlibrary{arrows.meta}
\usepackage[dvipsnames]{xcolor}

\begin{document}

\title{Computing the Critical Temperature of the Affine-Transformed $D=3$ Ising Model Using Masked Autoregressive Flow}

\author{Kai Svenson}
\affiliation{Theory Group, Department of Physics, The University of Texas at Austin, Austin, Texas 78712, USA}
\author{George T.\ Fleming}
\affiliation{Fermi National Accelerator Laboratory, Batavia, IL 60510, USA}
\author{Richard C.\ Brower}
\author{Nobuyuki Matsumoto}
\author{Rohan Misra}
\affiliation{Department of Physics, Boston University, Boston, MA 02215, USA}

\begin{abstract}
    The simple Ising model provides a rich environment to build and study lattice field theories. As part of an ongoing project to construct a conformal field theory (CFT) on an arbitrarily curved manifold~\cite{brower_intro,Brower_2021,ayyar2023operatorproductexpansionradial}, in this work we develop methods to measure the critical temperature $\beta_c$ of the affine-transformed Ising model on the face-centered cubic (FCC) lattice. The main challenge in this endeavor is finding a computationally efficient and accurate method of interpolating and extrapolating Monte Carlo observables with respect to coupling coefficients and temperature. Herein, we compare two such methods. A traditional statistical approach uses the multiple histogram (MH) method, while a newer machine learning approach uses a masked autoregressive flow (MAF) to estimate the underlying probability density function of a set of observables. While the MH method is specifically designed to interpolate and extrapolate Monte Carlo observables, we find that MAF is a viable alternative for measuring $\beta_c$ with a computational cost that scales more favorably. Furthermore, we comment on additional advantages of MAF relevant to our work, such as extrapolating in system volume.
\end{abstract}

\maketitle

\section{Introduction}

One of the great unsolved problems in physics emerging out of the twentieth century is the formulation of a self-consistent theoretical framework for harmonizing general relativity with the Standard Model of particle physics, which should describe the interaction between dynamically fluctuating curved spacetime manifolds and fluctuating quantum fields that live on those manifolds.  One theoretical obstruction to constructing such a framework is how to properly formulate a general non-perturbative quantum field theory on a static curved manifold.  The limit of weakly-coupled quantum field theories on nearly-flat manifolds has been fairly well understood for decades \cite{Jack:1983sk}. There has also been very exciting progress in the path integral formulation of the theory of quantum gravity (see, for example, Refs.~\cite{Dai:2023tud, Dai:2024vjc}) without quantum fields on the manifold.

Some progress has been made in the past decade or so extending beyond the perturbative regime by focusing on conformal field theories (CFT) in radial quantization, where the static $D$-dimensional spacetime manifold is $\mathbb{R} \times \mathbb{S}^{D-1}$.
The three-dimensional Ising CFT serves as a key benchmark in this program. Its scaling dimensions and OPE coefficients are now known to high precision from the conformal bootstrap~\cite{ElShowk2014SolvingT3, Kos2014BootstrappingMC, Kos:2016ysd}, providing exact targets against which any lattice construction can be compared. Two complementary lattice approaches aim to realize this CFT directly in radial quantization on $\mathbb{R}\times\mathbb{S}^2$: the fuzzy sphere regularization, which discretizes the sphere using electrons in the lowest Landau level and has reproduced the state-operator correspondence and OPE data of the transition~\cite{Zhu2022UncoveringCS, Hu2023OperatorPE}; and the \textit{quantum finite elements} (QFE) program, which uses a simplicial Regge triangulation of the sphere~\cite{Brower_2021, ayyar2023operatorproductexpansionradial}.

A central lesson of the QFE construction of the Ising model on $\mathbb{S}^2$~\cite{brower2024isingmodelmathbbs2} is that reaching the correct CFT in the continuum limit requires a precise map between the lattice couplings and the target geometry. In two dimensions this map was found analytically~\cite{Brower_2023_2d_ising}: restoring rotational symmetry of the critical two-point function corresponds to an affine transformation taking circles to ellipses. The \textit{affine conjecture}~\cite{Brower:2025oti} posits that this generalizes to higher dimensions and arbitrarily curved manifolds — that there exists a smooth map between the affine lattice couplings and the target geometry such that, tuned to criticality, the lattice theory recovers the desired CFT in the continuum. Testing this conjecture requires knowing the critical temperature $\beta_c$ as a function of the couplings $\{K_i\}$, precisely the quantity we set out to measure in this work.

While QFE is a perturbative method, in this work we make progress toward a generalization of QFE which will apply non-perturbatively to a wider class of manifolds. As with QFE, we are interested in the construction of a $\phi^4$ theory on an arbitrarily curved manifold at the Wilson-Fisher fixed point. Locally, this theory may be seen as the continuum limit of the critical Ising model in $\R^3$ under an affine transformation~\cite{ayyar2023operatorproductexpansionradial}. Since this theory is conformal at its critical temperature, it is entirely characterized by a set of well-known scaling dimensions and operator product expansion (OPE) coefficients~\cite{ayyar2023operatorproductexpansionradial}. As a result, by choosing to study conformal field theories (CFTs), we have a convenient method of checking that our construction preserves desired symmetries by comparing its scaling dimensions and OPE coefficients to their known values.

At this stage in the project, we are focused on constructing the action of the affine-transformed Ising model:
\begin{align} \label{eq:ising_action}
    S &= -\beta \sum_{\vb{x}} \sum_{i = 1}^{z/2} K_{i} s_{\vb{x}} s_{\vb{x} + \vu{e}_i}
\end{align}
$\sum_{\vb{x}}$ runs over all lattice sites, and $i$ is an index that runs over unique links to nearest neighbors (we sum up to half the coordination number $z$ to prevent double-counting). See \cref{fig:fcc} for a visual. The $s_{\vb{x}} \in \{-1, 1\}$ constitute the spin degrees of freedom and $\vu{e}_i$ is defined such that $\vb{x}$ and $\vb{x} + \vu{e}_i$ form a nearest-neighbor pair.

The affine transformation is realized through the set of tunable dimensionless coupling coefficients $\{K_i\}$, which are a function of the curvature of a given manifold. Since our overarching objective is a construction that applies to an arbitrarily curved manifold, it remains to find the critical temperature $\beta_c$ as a function of $\{K_i\}$, which we evaluate numerically. For a fixed set of $\{K_i\}$, we use Monte Carlo methods to measure the energy and magnetization variances as functions of $\beta$. As the energy variance is related to the specific heat, and the magnetization variance to the susceptibility, both are expected to peak sharply (diverge in the thermodynamic limit) at the critical temperature, which is how we measure $\beta_c$.

However, Monte Carlo simulations of large-volume Ising models are computationally very expensive. Hence, running iterated simulations to increase the precision on $\beta_c$ is simply impractical. The alternative then is to find a method of interpolating $\beta_c$ from a set of Monte Carlo measurements taken in the neighborhood of $\beta_c$. In this work, we compare two methods. The first is the \textit{multiple histogram} (MH) method~\cite{newman}, which reweights samples of Monte Carlo observables taken at various temperatures to yield a new effective sample taken at a unique temperature. The second method uses \textit{Masked Autoregressive Flow} (MAF)~\cite{papamakarios2018maskedautoregressiveflowdensity} to model the underlying probability density function (PDF) of the Monte Carlo observables. The focus of this work will be comparing the performance of these methods, their strengths, and weaknesses. First, we discuss some of the preliminary steps that are necessary to complete this analysis.

\section{Monte Carlo Methods}

Using standard Monte Carlo methods~\cite{newman} and the computational resources offered by the Fermilab Lattice QCD Facility, we are able to measure thermal averages and fluctuations of observables governed by the Ising action in \cref{eq:ising_action}. All of the results shared in this work are done on the face-centered-cubic (FCC) lattice with $V = 32 \times 32 \times 32 = 32768$ spins using the Swendsen-Wang algorithm~\cite{sw_algorithm} and periodic boundary conditions (see \cref{fig:fcc}). In particular, we are interested in the fluctuations of the energy along each lattice edge and the magnetization, since they each peak sharply at criticality, and hence provide separate estimates of $\beta_c$. The spread in these measurements provides a useful reference for finite-size effects since all measurements of $\beta_c$ should agree in the infinite-volume limit. To make these terms more precise, the energy $E_i$ associated with each lattice edge $i$ is defined as:
\begin{align}
    E_i &= -\sum_{\vb{x}} s_{\vb{x}} s_{\vb{x} + \vu{e}_i}
\end{align}
At times, it will be notationally convenient to assemble the $\{E_i\}$ and $\{K_i\}$ into $z/2$-component vectors: $\vec{E}$ and $\vec{K}$. For the FCC lattice, $z/2 = 6$. The total energy is then $U = \vec{K}\cdot\vec{E}$. The energy variance along edge $i$ is defined as $\Var(E_i) = \expval{E_i^2} - \expval{E_i}^2$. Unless explicitly stated otherwise, angled brackets $\expval{\cdot}$ and $\Var(\cdot)$ denote thermal averages and variances. The magnetization takes the standard definition:
\begin{align}
    m &= \sum_{\vb{x}} s_{\vb{x}}
\end{align}

\begin{figure}
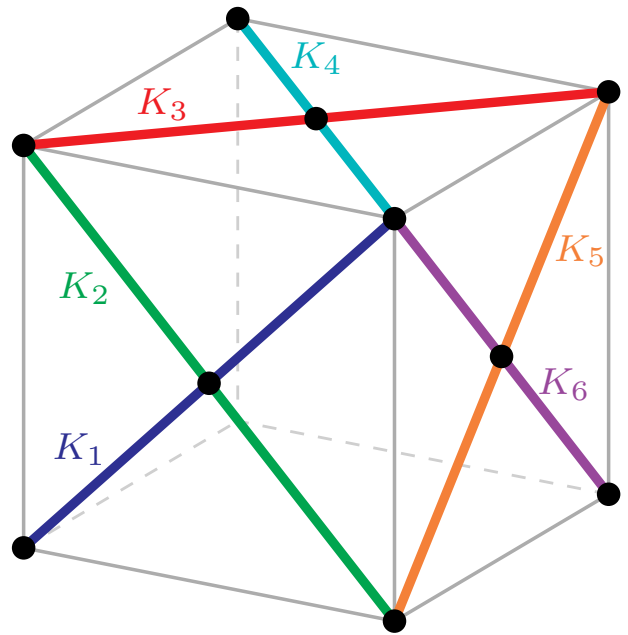

    \centering
    \includestandalone[width=\columnwidth]{fcc}
    \caption{FCC lattice unit cell with labeled coupling coefficients. Spins are located on the black nodes. The coordination number of the lattice is $z=12$. Some bonds are suppressed for clarity.}
    \label{fig:fcc}
\end{figure}

The simplest case we can consider is measuring $\Var(E_1/V)$, $\dots$, $\Var(E_6/V)$, and $\Var(|m|/V)$ as functions of $\beta$ and $K_i$, while fixing all other $K_{j \neq i} = 1$. The preliminary results are shown in \cref{fig:prelim_results}, where we choose $i=6$. Each plot demonstrates the expected behavior, clearly distinguishing the ferromagnetic phase in the upper right and the paramagnetic phase in the lower left. As predicted, \cref{fig:prelim_a,fig:prelim_b,fig:prelim_c} each display a sharp peak in $\Var(E_5/V)$, $\Var(E_6/V)$, and $\Var(|m|/V)$, respectively, which we can fit to measure $\beta_c$ (as described in \cref{sec:fit_beta_c}). Furthermore, for $K_6=1$, the location of the peak is consistent with $\beta_c = \SI{0.1020707(2)}{}$ calculated in previous work~\cite{betac_YU201575}.
\begin{figure*}
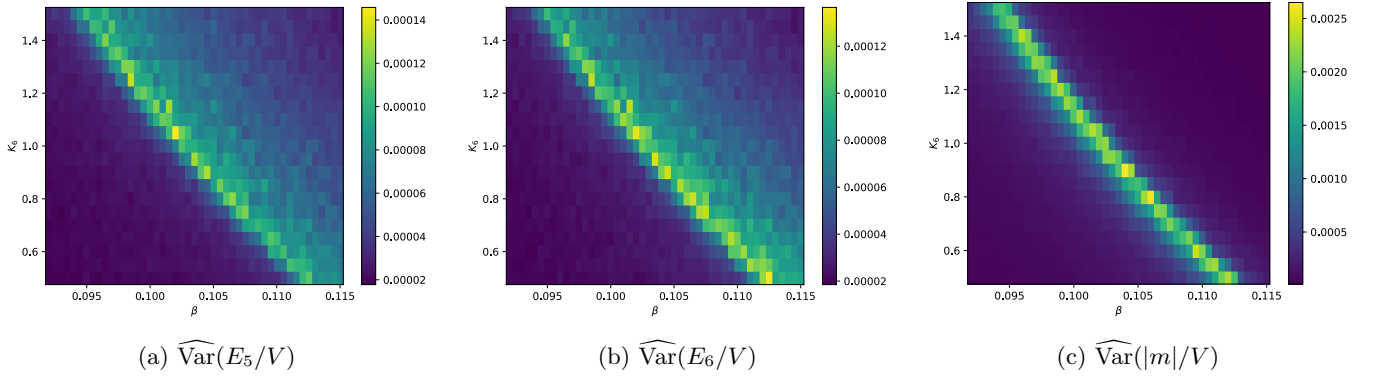

    \centering
    \begin{subfigure}[t]{0.32\textwidth}
        \centering
        \includesvg[width=\textwidth]{figs_150824_sw_coarse/k5_energy_var.svg}
        \caption{$\widehat{\Var}(E_5/V)$}
        \label{fig:prelim_a}
    \end{subfigure}
    \hfill
    \begin{subfigure}[t]{0.32\textwidth}
        \centering
        \includesvg[width=\textwidth]{figs_150824_sw_coarse/k6_energy_var.svg}
        \caption{$\widehat{\Var}(E_6/V)$}
        \label{fig:prelim_b}
    \end{subfigure}
    \hfill
    \begin{subfigure}[t]{0.32\textwidth}
        \centering
        \includesvg[width=\textwidth]{figs_150824_sw_coarse/magnetization_var.svg}
        \caption{$\widehat{\Var}(|m|/V)$}
        \label{fig:prelim_c}
    \end{subfigure}
    \caption{Various observables of interest plotted against $\beta$ and $K_6$. $V = 32^3$ is the total number of spins. All other coupling coefficients are fixed to one: $K_{i\neq 6} = 1$. Each pixel is a sample variance taken from 1000 data points.}
    \label{fig:prelim_results}
\end{figure*}

However, for the purposes of maintaining conformal symmetry in the continuum limit, we require a more precise measurement of $\beta_c$ than what can be achieved with \cref{fig:prelim_results}. Simply performing a greater number of Monte Carlo simulations is computationally intractable, which leads us to the focus of our work: comparing methods of Monte Carlo interpolation and extrapolation.

\section{The Multiple Histogram Method}

The multiple histogram (MH) method~\cite{newman} is a general-purpose statistical method of interpolating and extrapolating Monte Carlo observables given an ensemble of random samples taken over a range of action parameters. The MH method can be immediately applied toward our goal of increasing the resolution of \cref{fig:prelim_results} in $\beta$ and all the $K_i$, enabling a more precise estimate of $\beta_c$. Notably, we cannot use the MH method to extrapolate observables in system volume. Such a capability would be highly desirable when taking the infinite-volume limit of the Ising action, which motivates our search for alternative interpolation/extrapolation methods.

The details of the MH method can be found in~\cite{newman}, but to provide enough context for this analysis, suppose we have an observable $\obs$, and we would like to calculate its thermal average $\expval*{\obs(\beta, \vec{K})}$. To do this, we perform $m$ Monte Carlo simulations, each with a different set of parameters $\{\beta_j, \vec{K}_j\}$ with $j \in \qty{1, \dots, m}$. In the $j$th simulation, we take $n_j$ samples of $\obs$ and $\vec{E}$, which we label $\qty{\obs_{j,s}}$ and $\{\vec{E}_{j, s}\}$ with  $s \in \qty{1, \dots, n_j}$. At its core, the multiple histogram method uses an iterative self-consistency process to estimate the partition function $\widehat{Z}(\beta, \vec{K})$. We use hats to denote variables that are estimates as opposed to exact. With $\widehat{Z}$, the samples $\qty{\obs_{j,s}}$ can be reweighted to obtain an estimate of $\expval*{\obs(\beta', \vec{K}')}$ at an unsimulated parameter set $\{\beta', \vec{K}'\}$. Let:
\begin{subequations}\label{eq:reweighting_factors}
    \begin{align}
        \widehat{p}_{j, s} &= \frac{\exp(-\beta' \vec{K}' \cdot \vec{E}_{j,s})}{\widehat{Z}(\beta', \vec{K}')} \label{eq:pjs}
        \\
        \widehat{q}_{j, s} &= \sum_{k=1}^m \frac{n_k}{N} \cdot \frac{\exp(-\beta_k \vec{K}_k \cdot \vec{E}_{j,s})}{\widehat{Z}(\beta_k, \vec{K}_k)}\label{eq:qjs}
    \end{align}
\end{subequations}
with $N = \sum_{j=1}^m n_j$. Then, our estimate of $\expval*{\obs(\beta', \vec{K}')}$ is:
\begin{align}\label{eq:mh_est}
    \widehat{\obs} &= \frac{1}{N}\sum_{j=1}^m\sum_{s=1}^{n_j} \obs_{j,s} \frac{\widehat{p}_{j, s}}{\widehat{q}_{j, s}}
\end{align}

An estimator for $\Var(\obs(\beta', \vec{K}'))$ is:
\begin{align}\label{eq:mh_var_est}
    \widehat{\Var}\qty(\obs) &=
    \frac{N}{N-1}\qty(
    \widehat{\obs^2} - (\widehat{\obs})^2
    )
\end{align}

Justification for these estimators and intuition for our choice of notation are given in \cref{sec:mh_estimators}. For reliable estimates, a conservative bound for the temperature difference $\Delta \beta$ between consecutive simulations in the vicinity of $\beta$ is given by~\cite{newman}:
\begin{align}\label{eq:delta_beta_req}
    \qty(\frac{\Delta \beta}{2})^2 < \frac{1}{\Var(U)}
\end{align}
Across all simulations, $\Delta \beta = \SI{5e-4}{}$, and the maximum recorded estimate for $\Var(U)$ is $\sim \SI{5.2e6}{}$. As intended, our choice of $\Delta \beta$ comfortably satisfies \cref{eq:delta_beta_req} by roughly a factor of 3. This inequality ensures that the energy distributions of consecutive simulations sufficiently overlap with each other, and prevents the reweighting factors $\widehat{p}_{j, s} / \widehat{q}_{j, s}$ from becoming exponentially suppressed. As a double-check, we verify this graphically in \cref{fig:mh_reliable}. \cref{fig:eng_hist} plots the empirical energy density distribution for the $K_6=1$ slice. The continuous support between the lowest and highest recorded energy densities enables the MH method to comfortably interpolate any observable within the corresponding temperature range. Additionally, \cref{fig:mh_ess} plots the effective sample size (ESS), defined as follows. First, let $w_{j, s} = \frac{\widehat{p}_{j, s}}{\widehat{q}_{j, s}}$. Then:
\begin{align}
    \text{ESS} &= \frac{\qty(\sum_{j=1}^m\sum_{s=1}^{n_j} w_{j, s})^2}{\sum_{j=1}^m\sum_{s=1}^{n_j} w_{j, s}^2}
\end{align}
We acknowledge the caveats of formally using such a definition for the ESS. However, for our purposes, it serves as a heuristic confirmation that the weights in our parameter space of interest are non-vanishing.
\begin{figure*}
    \centering
    \begin{subfigure}[t]{0.49\textwidth}
        \includegraphics[width=\textwidth]{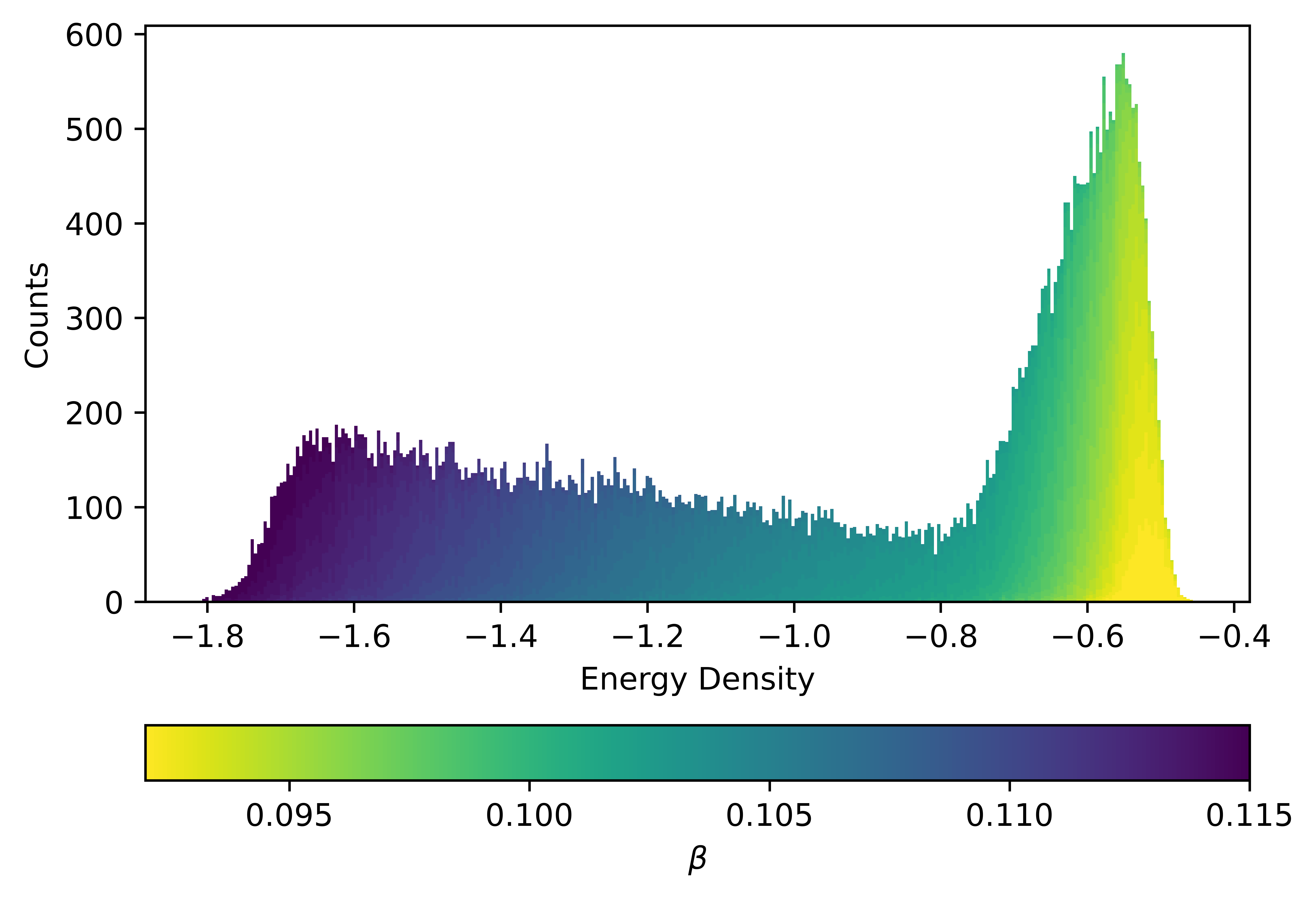}
        \caption{Empirical energy density ($U / V$) distribution for $K_1, \dots, K_6 = 1$. Color denotes the temperature at which the energy is recorded.}
        \label{fig:eng_hist}
    \end{subfigure}
    \hfill
    \begin{subfigure}[t]{0.49\textwidth}
        \centering
        \includesvg[width=\textwidth]{figs_150824_sw_coarse/multi_hist_log_ess.svg}
        \caption{Effective Sample Size. The white box separates the interpolation and extrapolation regions.}
        \label{fig:mh_ess}
    \end{subfigure}
    \caption{Empirical energy density and effective sample size used to verify that the multiple histogram method yields reliable estimates in the parameter space of interest.}
    \label{fig:mh_reliable}
\end{figure*}

After applying the MH method to the raw data presented in \cref{fig:prelim_results}, we arrive at the results in \cref{fig:mh_results}. As expected, the MH method is clearly able to interpolate the magnetization and directional energy variances observed in \cref{fig:prelim_results}, with extrapolation failing on the order of $\Delta \beta$ outside of the simulated region. Although less expensive than iterated Monte Carlo simulations, the MH method still comes at a significant computational cost, which does not scale well if we consider measuring the whole parameter space $\{\beta, \vec{K}\}$. Additionally, \cref{eq:delta_beta_req} implies that $\Delta \beta$ should be made exceptionally small near $\beta_c$ where $\Var(U)$ is expected to peak sharply, further raising the cost of the MH method. These computational limitations, as well as the MH method's inability to extrapolate in system volume, lead us to consider alternative methods of estimation, which brings us to Masked Autoregressive Flow.
\begin{figure*}
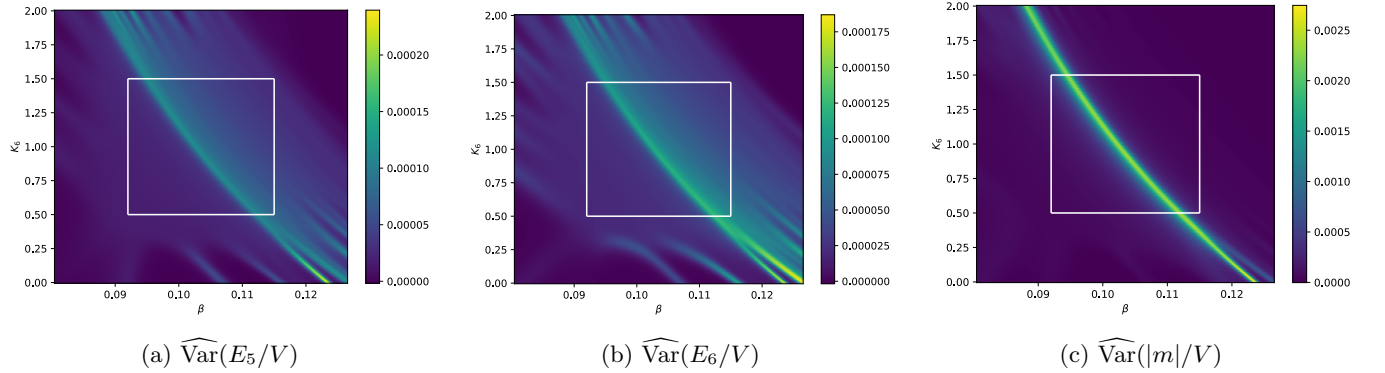

    \centering
    \begin{subfigure}[t]{0.32\textwidth}
        \centering
        \includesvg[width=\textwidth]{figs_150824_sw_coarse/multi_hist_k5_energy_var.svg}
        \caption{$\widehat{\Var}(E_5 / V)$}
        \label{fig:mh_eng_5_var}
    \end{subfigure}
    \hfill
    \begin{subfigure}[t]{0.32\textwidth}
        \centering
        \includesvg[width=\columnwidth]{figs_150824_sw_coarse/multi_hist_k6_energy_var.svg}
        \caption{$\widehat{\Var}(E_6 / V)$}
        \label{fig:mh_eng_6_var}
    \end{subfigure}
    \hfill
    \begin{subfigure}[t]{0.32\textwidth}
        \centering
        \includesvg[width=\columnwidth]{figs_150824_sw_coarse/multi_hist_magnetization_var.svg}
        \caption{$\widehat{\Var}(|m| / V)$}
        \label{fig:mh_mag_var}
    \end{subfigure}
    \caption{$\Var(E_5/V)$, $\Var(E_6/V)$, and $\Var(|m|/V)$ estimated with the MH method. These are the same observables from \cref{fig:prelim_results} interpolated and extrapolated to have five times the resolution in $\beta$ and $K_6$. The white boxes separate the interpolation and extrapolation regions.}
    \label{fig:mh_results}
\end{figure*}

\section{Masked Autoregressive Flow}

Masked Autoregressive Flow~\cite{papamakarios2018maskedautoregressiveflowdensity} (MAF) is a machine learning approach to the problem of estimating the underlying probability density function (PDF) of observables. The model parameterizes the target multi-dimensional PDF as a product of single-dimensional PDFs. This factorization is not unique, which is a freedom we can later leverage as a method of error quantification (see \cref{sec:maf_uncertainty}). For example, two ways the factorization can be performed are:
\begin{align}\label{eq:maf_single_pdfs}
    \begin{split}
        &\pdf(m, \vec{E})
        \\
        &= \pdf(m) \cdot \qty[\prod_{i=1}^6 \pdf(E_i|m, E_1, \dots, E_{i-1})]
        \\
        &=\qty[\prod_{i=1}^6\pdf(E_i|E_1, \dots, E_{i-1})] \cdot \pdf(m | \vec{E})
    \end{split}
\end{align}
(we have suppressed the conditioning on $\beta$ and $\vec{K}$ for clarity). The single-dimensional PDFs are then estimated using a series of layers with masked connections. The masks are chosen carefully so as to preserve the proper conditioning between the single-dimensional PDFs. See the original paper~\cite{papamakarios2018maskedautoregressiveflowdensity} for a detailed description of the network architecture. With the final PDF, all moments of the distribution can be computed by directly sampling from the learned PDF.

Karsch et al. \cite{karsch2022machinelearningapproachclassification} found that MAF overfits less than the MH method in many-flavor QCD simulations, suggesting that the machine learning approach may be an improvement on the traditional statistical method in our context as well. Furthermore, MAF can easily be extended to extrapolate observables in system volume, which is a key requirement in our work as mentioned previously. As a result, here, we have implemented MAF to learn the joint PDF of $m$ and $\vec{E}$ conditioned on $\beta$ and $\vec{K}$: $\pdf(m, \vec{E} | \beta, \vec{K})$.

To maximize portability and expandability, our implementation is built using the PyTorch Python library. We build on the implementation of Zhihan Yang~\cite{maf_pytorch_trunk}, adding the functionality to learn conditional PDFs. After training on the data summarized in \cref{fig:prelim_results}, we arrive at the results in \cref{fig:maf}. Over the course of our experimentation, we made several key modifications to the stock MAF implementation in order to adequately reproduce the results of \cref{fig:mh_results}. These are discussed in the following sections.

\begin{figure*}
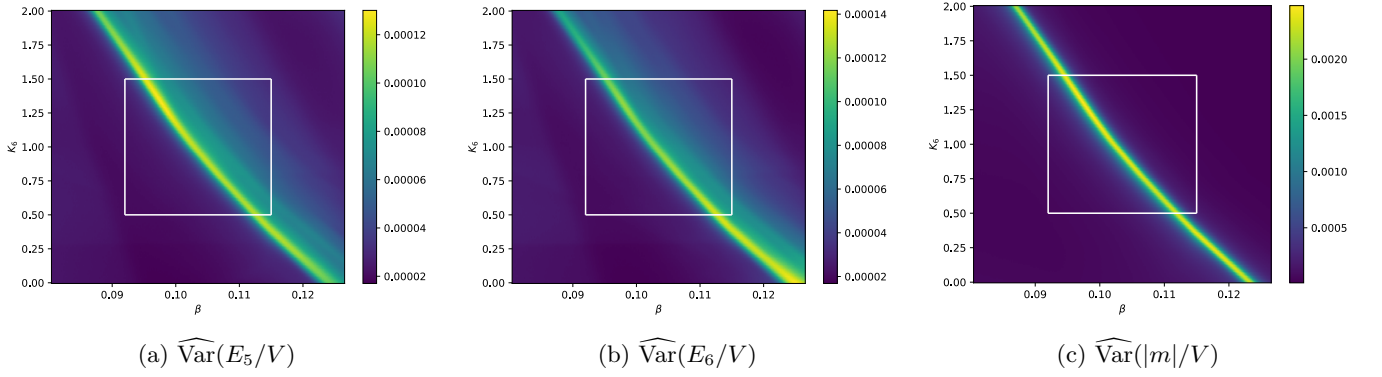

    \centering
    \begin{subfigure}[t]{0.32\textwidth}
        \centering
        \includesvg[width=\textwidth]{maf/dist_221225_ar2_h2x100_bn0_e487_mh_eng_dir5.svg}
        \caption{$\widehat{\Var}(E_5/V)$}
        \label{fig:maf_eng_5}
    \end{subfigure}
    \hfill
    \begin{subfigure}[t]{0.32\textwidth}
        \centering
        \includesvg[width=\textwidth]{maf/dist_221225_ar2_h2x100_bn0_e487_mh_eng_dir6.svg}
        \caption{$\widehat{\Var}(E_6/V)$}
        \label{fig:maf_eng_6}
    \end{subfigure}
    \hfill
    \begin{subfigure}[t]{0.32\textwidth}
        \centering
        \includesvg[width=\textwidth]{maf/dist_221225_ar2_h2x100_bn0_e487_mh_mag.svg}
        \caption{$\widehat{\Var}(|m|/V)$}
        \label{fig:maf_mag}
    \end{subfigure}
    \caption{$\Var(E_5/V)$, $\Var(E_6/V)$, and $\Var(|m|/V)$ estimated with MAF. This implementation contained two autoregressive layers in series, each with two hidden layers, and each hidden layer having 100 nodes. No batch normalization layers are used. These plots have the same resolution as~\cref{fig:mh_results}. The white boxes separate the interpolation and extrapolation regions.}
    \label{fig:maf}
\end{figure*}

\subsection{Batch Normalization in MAF}

Batch normalization~\cite{ioffe2015batchnormalizationacceleratingdeep} is a tool widely used in neural networks that is generally known to improve runtime and numerical stability (although the reasons for this are ``still poorly understood"~\cite{santurkar2019doesbatchnormalizationhelp}). During training, a batch normalization layer will perform a linear transformation on a given batch so that the output has zero mean and unit variance: it subtracts off the mean and divides by the standard deviation (std.). During inference, the batch normalization parameters become fixed, although the choice of what to fix them to varies between implementations. Some examples include the mean and std. of the entire training set~\cite{papamakarios2018maskedautoregressiveflowdensity}, the mean and std. averaged over minibatches~\cite{ioffe2015batchnormalizationacceleratingdeep}, and a running mean and std. recorded during training~\cite{dinh2017densityestimationusingreal,pytorch_batch_norm}. However, for our implementation, any of these approaches would require averaging samples taken from different phases, which results in statistics that are representative of neither phase, not to mention physically irrelevant. For example, when transitioning from the paramagnetic phase to the ferromagnetic phase, $m$ transitions from a unimodal to a bimodal distribution. During batch normalization training, samples of $m$ from different configurations $\{\beta_j, \vec{K}_j\}$ are pooled together and averaged. Conversely, during inference, we sample $m$ from a single configuration $\{\beta', \vec{K}'\}$. In our experiments, this mismatch results in extremely poor model performance since the behavior of $m$ is highly sensitive to $\beta$, and cannot be normalized by the global, $\beta$-independent statistics of batch normalization.

To correct this issue, the plots in \cref{fig:maf} are made without any batch normalization layers between the autoregressive layers. This choice comes at a cost, however. As anticipated by~\cite{papamakarios2018maskedautoregressiveflowdensity}, the absence of batch normalization results in numerical instability in larger networks and slower training times. While our implementation may have succeeded with a relatively small network consisting of two autoregressive layers, a larger model will be necessary to estimate the PDF over the full space spanned by $\{\beta, \vec{K}\}$. ``Conditional batch normalization"~\cite{devries2017modulatingearlyvisualprocessing}, a technique developed for the task of visual processing, may provide one solution. Further experimentation will be needed to find out whether the technique can be readily applied to PDF estimation.

\subsection{Base Distribution Symmetrization}

The base distribution which MAF transforms into the target PDF is a product of independent single-dimensional PDFs. These base single-dimensional PDFs are made by summing $N_G$ Gaussian distributions with independent means and variances $\{\mu_\ell, \sigma^2_\ell\}$, $\ell \in \qty{1, \dots, N_G}$ that are learned during training. Each Gaussian is weighted by a corresponding ``mixing coefficient" $\omega_\ell$ with $\sum_{\ell=1}^{N_G} \omega_\ell = 1$, also learned during training. We fix $N_G = 2$ since the magnetization in the ferromagnetic phase is bimodal. The stock implementation of MAF leaves $\mu_{1,2}$, $\sigma^2_{1,2}$, and $\omega_{1, 2}$ as independent training parameters, but we apply extra restrictions that reflect the $\Z_2$ symmetry of the Ising action. In particular, we force $\sigma^2_1 = \sigma^2_2$ and $\omega_1 = \omega_2 = \frac{1}{2}$, which we found to stabilize training. We do not force $\mu_1 = -\mu_2$ since we found that this degraded the model's ability to pinpoint the peaks in observable variances. Although the target PDF for $m$ is bimodal and symmetric about $m=0$, the base distribution is abstracted behind several autoregressive layers, and need not be centered on the origin. We also symmetrize the data set used to train the MAF network, as described in \cref{sec:maf_uncertainty}.

\section{Comparing MH and MAF}\label{sec:fit_beta_c}

Now that each method has been introduced, we present our final results: comparing the performance of the MH and MAF methods in estimating $\beta_c$ as a function of $\vec{K}$. For a fixed value of $K_6$, we fit the variance of an observable with a `generalized Lorentzian':
\begin{equation}\label{eq:fit_beta}
    a_1 \qty(1 + \qty|\frac{\beta-\beta_c}{a_2}|^{a_3})^{-1}
\end{equation}
where $\beta_c$, $a_1$, $a_2$, and $a_3$ are all fit parameters. The model is only required to find the peak in variance, and hence is only deployed in a small window around the largest recorded variance. A more complex model to take into account possible asymmetry about the peak could be used to measure other interesting quantities such as critical exponents, but was not needed here.

Repeating the fit for all desired values of $K_6$, we arrive at the comparison in \cref{fig:comparison}.
\begin{figure*}
    \centering
    \begin{subfigure}[t]{0.49\textwidth}
        \centering
        \includesvg[width=\textwidth]{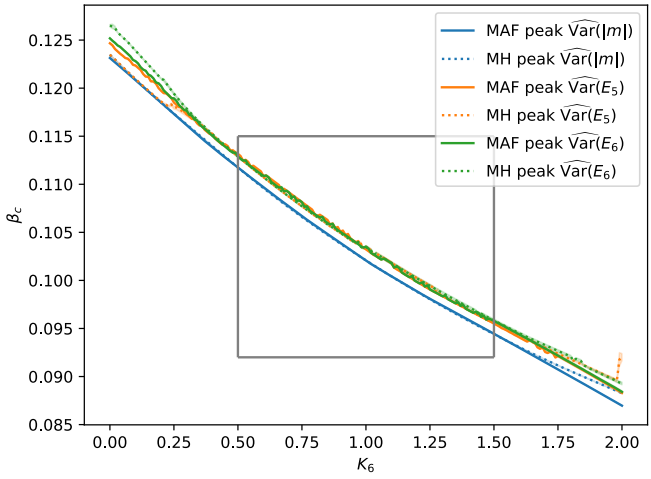}
        \caption{$\beta_c$ determined by fitting observables from \cref{fig:mh_results,fig:maf} with \cref{eq:fit_beta} for a given value of $K_6$.}
        \label{fig:compare_a}
    \end{subfigure}
    \hfill
    \begin{subfigure}[t]{0.49\textwidth}
        \centering
        \includesvg[width=\textwidth]{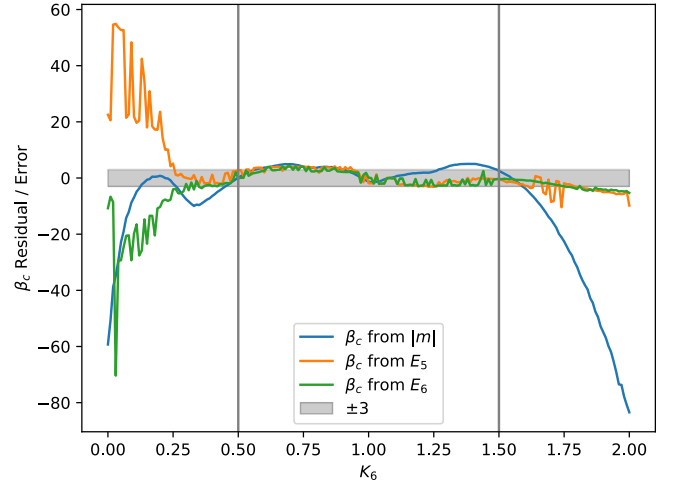}
        \caption{Residual plot between the MH and MAF methods, normalized by the fitting errors added in quadrature.}
        \label{fig:compare_b}
    \end{subfigure}
    \caption{Comparison of $\beta_c$ estimated with the MH and MAF methods. Error bars are derived from the fitting model (\cref{eq:fit_beta}) and are shown in (a), but are as narrow as the line widths. The grey lines separate the interpolation and extrapolation regions.}
    \label{fig:comparison}
\end{figure*}

The estimators are in near agreement with each other inside the interpolation region, and diverge outside the interpolation region on the scale of $\Delta \beta$. MAF is able to recover the estimate made by the MH method and, in our implementation, has a time complexity that scales much more slowly with a finer interpolation grid. For example, to increase the resolution of \cref{fig:mh_results} by another factor of five would require an additional pass over all $N$ samples for each of 25 times as many grid points (a factor of five each for $\beta$ and $K_6$). This scaling can be somewhat mitigated by only performing a refinement on pixels near $\beta_c$. However, given that a single pass is already an expensive computation requiring HPC resources, this reduced scaling is still undesirable. Meanwhile, for the same given set of Monte Carlo data, MAF only needs to be trained once, and can be queried for any interpolated $\beta$ without great additional cost.

\section{Uncertainty Quantification}

Any measurement of $\beta_c$ must also come with an uncertainty, which is itself a nuanced quantity dependent on our choice of implementation. Here we discuss the main sources of uncertainty associated with each method.

\subsection{MH Uncertainties}\label{sec:mh_err}

The largest source of systematic uncertainty in the MH method is associated with the `reweighting factors': $\hat{p}_{j, s}$ and $\hat{q}_{j, s}$ (defined in \cref{eq:reweighting_factors}). The Boltzmann factors are exact, so it is specifically our estimate of $\widehat{Z}$ that carries any potential source of bias. As described in~\cite{newman}, $\widehat{Z}$ is derived from an iterative and convergent loop, meaning its error can be made arbitrarily small compared to statistical uncertainties. Furthermore, as discussed in \cref{sec:mh_estimators}, the MH method can be viewed as ``self-normalized importance sampling". Any bias in $\widehat{\obs}$ derived from $\widehat{Z}$ is known to vanish asymptotically in the infinite sample limit~\cite{mcbook}.

Statistical uncertainty in the MH method for averages can be estimated directly via \cref{eq:mh_sem_est}. Of course, we are more interested in the uncertainties on variances, which do not have a similarly simple estimator. Instead, we may consider jackknife methods, or if we only concern ourselves with $\beta_c$, the uncertainty derived from our fitting model, \cref{eq:fit_beta}. It is also worth noting that at our lattice volume and number of samples, finite-size effects are comparable to statistical uncertainty. As seen in \cref{fig:compare_a}, there is a significant difference between $\beta_c$ measured with $\widehat{\Var}(|m|)$ and $\widehat{\Var}(E_{5,6})$. As a result, if our goal is to measure $\beta_c$ at infinite volume, statistical uncertainty is not the limiting factor.

One final remark regarding uncertainty quantification in the MH method is that we have a convenient `check' for determining how far we can reliably extrapolate. Consider using the MH method to extrapolate in $K_5$ while fixing all $K_{i \neq 5} = 1$. Due to the symmetry of the FCC lattice, this extrapolation should result in an estimate of $\beta_c$ in agreement with the interpolated estimate using $K_6$. This can be seen in \cref{fig:k5_extrap}, where the two estimates are in near agreement close to $K_1, \dots, K_6 = 1$, and begin to diverge at the expected scale of $\Delta \beta$.
\begin{figure}
    \centering
    \includesvg[width=\linewidth]{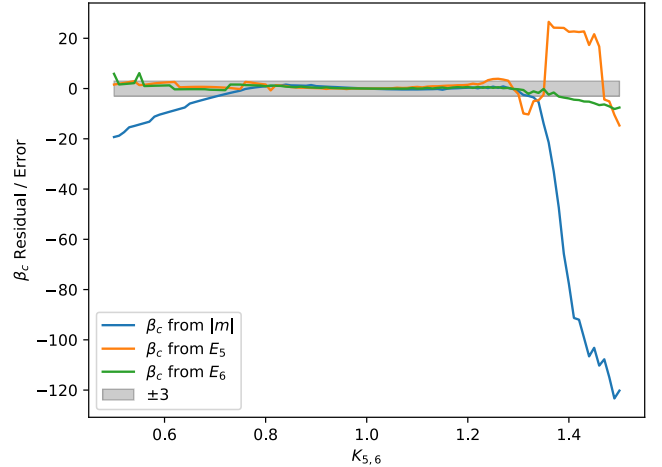}
    \caption{Residual plot between the MH method interpolating $\beta_c$ along $K_6$, and the MH method extrapolating $\beta_c$ along $K_5$, normalized by the fitting errors added in quadrature.}
    \label{fig:k5_extrap}
\end{figure}

\subsection{MAF Uncertainties}\label{sec:maf_uncertainty}

With MAF, systematic uncertainty can be quantified by varying the training process and measuring the spread in final results. For example, we observed that the batching and ordering of the training data can have a significant impact on model performance, especially when batch normalization is included. Additionally, the estimated PDF can be factorized using any permutation of the observables. \cref{eq:maf_single_pdfs} gives two examples. The different factorizations of $\pdf(m, \vec{E})$ correspond to different network architectures. In principle, these architectures should behave similarly, but as noted by Papamakarios et al., ``a drawback of autoregressive models is that they are sensitive to the order of the variables"~\cite{papamakarios2018maskedautoregressiveflowdensity}. In future work, we can train the model several times using different permutations and leverage this sensitivity to estimate the systematic uncertainty in MAF.

Uncertainties derived from training in MAF can be handled using standard practices in machine learning, in particular by reserving a fraction of the original Monte Carlo data for testing to prevent overfitting. In our experiments, of the 1000 samples that went into each pixel of \cref{fig:prelim_results}, 500 went to training, and 500 went to testing. The training data set is also symmetrized. For each sample $\{m, \vec{E}\}$, we augmented the training set with $\{-m, \vec{E}\}$ in order to reduce uncertainty and enforce the $\Z_2$ symmetry in the Ising action.

\section{Conclusions}

Constructing a CFT on an arbitrarily curved manifold via the affine conjecture requires knowing the critical temperature $\beta_c$ as a continuous function of the coupling coefficients $\{K_i\}$ in the infinite volume limit, where the $\{K_i\}$ could deviate from unity by as much as 50\%. Then, at criticality, the ratios of distances can be determined from correlation functions, which can be tuned to match the target geometry. In principle, this is an extremely computationally intensive task. The central goal of this work was to identify a method capable of interpolating and extrapolating Monte Carlo observables both in couplings \textbf{and} in the volume that is both sufficiently accurate and computationally efficient to map out $\beta_c$ across the parameter space. The MH method would be sufficient to map the parameter space of couplings, but only at a fixed volume. The masked autoregressive flow (MAF) could be used to interpolate and extrapolate in the volume if it could first be verified that MAF works as well as MH on a fixed volume. 

Our main result is this verification. For most of the interpolation region, MAF reproduces the $\beta_c$ estimates of the MH method within $3\sigma$ (\cref{fig:comparison}), and does so with a computational cost that grows more slowly with more data or higher target precision. This agreement validates MAF as an effective interpolation tool in a setting where the MH method already serves as a trusted benchmark.

Establishing this agreement is the key step that justifies our future use of MAF in a regime where the MH method cannot follow. Unlike the MH method, MAF can interpolate between datasets computed at different lattice volumes, which is precisely what is needed to extrapolate observables to the infinite-volume limit. Since the MH method offers no analogous capability, the demonstrated agreement of the two methods at fixed volume gives us confidence to deploy MAF for the volume extrapolation required to take the continuum limit. Pursuing this extrapolation, together with the implementation challenges discussed above, most notably adapting batch normalization to our setting, constitutes the natural next stage of this program, bringing us closer to the broader goal of studying field theories on curved manifolds.

\begin{acknowledgments}
    Kai Svenson thanks Maria Spiropulu, Joseph Lykken, Mark Wise, and Dan Hackett for their mentorship and valuable conversations throughout this project. Kai Svenson also thanks the Caltech SURF program for helping sponsor this project.
    N.M. is supported by the Scientific Discovery through Advanced Computing (SciDAC) program under FOA LAB-2580, funded by the Department of Energy, Office of Science.
    This work is supported by the U.S. Department of Energy, Office of Science, Office of High Energy Physics, under award numbers DE-SC0011925 and DE-SC0019219.
    This document is prepared by the Quantum Finite Elements (QFE) collaboration using the resources of the Fermi National Accelerator Laboratory (Fermilab), a U.S. Department of Energy, Office of Science, Office of High Energy Physics HEP User Facility. Fermilab is managed by Fermi Forward Discovery Group, LLC, acting under Contract No. 89243024CSC000002.
\end{acknowledgments}

\bibliography{main}

\appendix

\section{Multiple Histogram Method Estimators}\label{sec:mh_estimators}

A complete and independent derivation of the estimator \cref{eq:mh_est} can be found in~\cite{newman}. Naturally, we would also like an estimator for the standard error of $\widehat{\obs}$ and $\Var(\obs)$. Although it is not directly presented as such in~\cite{newman}, the MH method is a type of importance sampling method, for which there already exist many sources~\cite{mcbook,GATZ19951185,Newman1999_error} with results that we can apply to our analysis.

To make the correspondence between the MH method and importance sampling, it will first be beneficial to consider a simplified case. Suppose we are interested in estimating some statistics of a random variable $X$ with PDF $p: \R \to \R$. However, suppose the only data given to us is a collection of i.i.d. samples $y_1, \dots, y_N \in \R$ from a different random variable $Y$ with PDF $q: \R \to \R$. Since we are dealing with two random variables, we have two notions of an expectation value. Given an observable $g: \R^N \to \R$ made from $N$ i.i.d. samples, we define:
\begin{subequations}
    \begin{align}
        \expval{g}_p &= \int_{\R^N} g(x_1, \dots, x_N) \prod_{i=1}^N p(x_i)\dd{x_i}
        \\
        \expval{g}_q &= \int_{\R^N} g(x_1, \dots, x_N) \prod_{i=1}^N q(x_i)\dd{x_i}
    \end{align}
\end{subequations}
Similar definitions follow for $\Var_{p, q}(g) = \expval{g^2}_{p, q} - \expval{g}_{p, q}^2$. Next, suppose we are interested in estimating the ``reweighted mean" $\expval{f}_p$ for some arbitrary observable $f: \R \to \R$. The most straightforward method of doing this is with the following estimator~\cite{mcbook}:
\begin{align}\label{eq:reweighted_mean}
    \widehat{f} &= \frac{1}{N} \sum_{i=1}^N f(y_i) w_i
\end{align}
with $w_i = p(y_i)/q(y_i)$. It is easy to demonstrate that $\widehat{f}$ is unbiased (\cref{prop:fhat_unbiased}) if $p$ and $q$ are known exactly. Of course, this is unrealistic in an experimental setting, but if we are able to obtain estimates of $p$ and $q$, then those can be substituted in our expression for $\widehat{f}$.

Next, we derive unbiased estimators for the squared standard error of $\widehat{f}$  and $\Var_p(f)$. Respectively:
\begin{align}
    \widehat{\Var}_q(\widehat{f}) &= \frac{1}{N-1}\qty(\frac{1}{N}\qty[\sum_{i=1}^N w_i^2 f(y_i)^2] - (\widehat{f})^2)\label{eq:unbiased_sem}
    \\
    \widehat{\Var}_p(f) &= \widehat{f^2} - (\widehat{f})^2 + \widehat{\Var}_q(\widehat{f})\label{eq:unbiased_var}
\end{align}
Derivations can be found in \cref{prop:unbiased_sem,prop:var_est}.

Finally, we apply this reasoning to the MH method. In this case, our given random sample is the entire set of $N$ samples $\{\obs_{j, s}\}$ with corresponding energies $\{\vec{E}_{j,s}\}$. This random sample is analogous to $y_1, \dots, y_N$. Consider any element of our sample set $\obs_{j, s}$, observed with energy $\vec{E}_{j, s}$. Let $\Omega_{j, s}$ be the number of microstates such that $\obs_{j, s}$ and $\vec{E}_{j, s}$ are simultaneously observed. Then, the probability $p_{j, s}$ that we could have observed $\obs_{j, s}$ and $\vec{E}_{j, s}$ at a configuration $\{\beta', \vec{K}'\}$ is:
\begin{align}
    p_{j,s} &= \Omega_{j, s} \cdot \frac{\exp(-\beta' \vec{K}' \cdot \vec{E}_{j, s})}{Z(\beta', \vec{K}')}
\end{align}
$p_{j,s}$ is analogous to $p(y_j)$, and is best estimated with $\Omega_{j, s} \widehat{p}_{j, s}$ as defined in \cref{eq:pjs}. We do not normally know $\Omega_{j, s}$, but as we shall soon see, this will not matter.

Next, notice that the probability $q_{j,s}$ that we actually observed $\obs_{j, s}$ and $\vec{E}_{j, s}$ simultaneously is:
\begin{align}
    q_{j,s} &= \sum_{k=1}^m \frac{n_k}{N} \cdot \Omega_{j, s} \cdot \frac{\exp(-\beta_k \vec{K}_k \cdot \vec{E}_{j, s})}{Z(\beta_k, \vec{K}_k)}
\end{align}
To break this down:
\begin{itemize}
    \item The probability that any observation $\obs$ came from the $k$th simulation is $n_k / N$. Notice, in order to apply the results of importance sampling to the MH method, we must treat all samples from all simulations as if they came from the same underlying distribution. Out of all $N$ samples, $n_k$ of them came from the $k$th simulation, so any randomly picked observation $\obs$ comes from simulation $k$ with probability $n_k/N$.
    
    \item The probability that we could have observed $\obs_{j, s}$ and $\vec{E}_{j, s}$, supposing that they were observed in the $k$th simulation, is $\Omega_{j, s} \exponential(-\beta_k \vec{K}_k \cdot \vec{E}_{j,s}) Z(\beta_k, \vec{K}_k)^{-1}$.
    
    \item The product of $\Omega_{j, s} \exponential(-\beta_k \vec{K}_k \cdot \vec{E}_{j,s}) Z(\beta_k, \vec{K}_k)^{-1}$ and $n_k / N$ is the probability that we could have measured $\obs_{j, s}$ in the $k$th simulation.
    
    \item By summing over all simulations, we get the total probability of observing $\obs_{j,s}$.
\end{itemize}
$q_{j,s}$ is analogous to $q(y_j)$, and is best estimated with $\Omega_{j, s} \widehat{q}_{j, s}$ as defined in \cref{eq:qjs}. Since we only need the ratio $p_{j,s} / q_{j,s}$, we do not need to estimate $\Omega_{j, s}$. Now, we may directly apply our estimators to the MH method. \Cref{eq:reweighted_mean} translates to \cref{eq:mh_est}. \Cref{eq:unbiased_sem} translates to:
\begin{align}\label{eq:mh_sem_est}
    \widehat{\Var}_\text{stat.}(\widehat{\obs}) &=
    \begin{aligned}[t]
        &
        \frac{1}{N(N-1)}\sum_{j, s} \qty(\frac{\widehat{p}_{j,s}}{\widehat{q}_{j,s}})^2 \obs_{j,s}^2
        \\
        -&
        \frac{1}{N-1} (\widehat{\obs})^2
    \end{aligned}
\end{align}
The `stat.' subscript emphasizes that this is not a thermal variance, but the squared standard error of $\widehat{\obs}(\beta)$. \Cref{eq:unbiased_var} translates to:
\begin{align}\label{eq:mh_var_est_prime}
    \widehat{\Var}'(\obs) &= \widehat{\obs^2} - (\widehat{\obs})^2 + \widehat{\Var}_\text{stat.}(\widehat{\obs})
\end{align}
The prime is to differentiate this estimator from \cref{eq:mh_var_est} (discussion below). We must emphasize that in taking the ratio $\widehat{p}_{j,s} / \widehat{q}_{j,s}$, we have introduced a bias since $\widehat{p}_{j,s}$ and $\widehat{q}_{j,s}$ are not independent. However, this bias is expected to vanish in the infinite sample limit~\cite{mcbook}.

Notably, \cref{eq:mh_var_est_prime} is not the estimator we have used in the main analysis. Against expectations, \cref{eq:mh_var_est_prime} did not yield any identifiable critical boundary. We suspect that although \cref{eq:mh_var_est_prime} may have reduced bias, it may also have a relatively larger variance through the factors of $\qty(\widehat{p}_{j,s} / \widehat{q}_{j,s})^2$. The estimator in \cref{eq:mh_var_est} is chosen to match the standard sample variance in the case of unweighted data. Significant discrepancies between reweighted estimators have also been observed by Gatz and Smith in~\cite{GATZ19951185}. Their comparison among three distinct estimators for the standard error of the reweighted mean revealed that there is no universally accepted `best' estimator for $\Var_\text{stat.}(\widehat{\obs}(\beta))$ (and $\Var(\obs(\beta))$ by extension). Explaining the discrepancy between all of the estimators in this section and in~\cite{GATZ19951185} will be a task for future work.

\begin{widetext}
\section{Derivations}\label{ap:derivations}

\begin{proposition}\label{prop:fhat_unbiased}
    \Cref{eq:reweighted_mean} is an unbiased estimator of $\expval*{f}_p$.
\end{proposition}
\begin{proof}
    \begin{subequations}
        \begin{align}
            \expval{\widehat{f}}_q &= \frac{1}{N} \sum_{i=1}^N \int_{\R^N}  f(y_i) \frac{p(y_i)}{q(y_i)} \prod_{j=1}^N q(y_j)\dd{y_j}
            \\
            &= \frac{1}{N} \sum_{i=1}^N \int_{\R}  f(y_i) p(y_i) \dd{y_i}
            \\
            &= \expval{f}_p
        \end{align}
    \end{subequations}
\end{proof}

\begin{proposition}\label{prop:unbiased_sem}
    \Cref{eq:unbiased_sem} is an unbiased estimator of $\Var_q(\widehat{f})$.
\end{proposition}
\begin{proof}
    First, we write $\Var_q(\widehat{f})$ in a convenient form.
    \begin{subequations}
        \begin{align}
            \Var_q(\widehat{f}) &= \expval{(\widehat{f})^2}_q - \expval{\widehat{f}}^2_q\label{eq:unbiased_sem_helper1}
            \\
            &= \frac{1}{N^2} \sum_{i, j=1}^N \int_{\R^N} f(y_i)f(y_j) \frac{p(y_i)p(y_j)}{q(y_i)q(y_j)} \prod_{k=1}^N q(y_k) \dd{y_k} - \expval{f}^2_p
            \intertext{Let $w = p/q$. We split the summation into two cases: $i=j$ and $i\neq j$.}
            &= \frac{1}{N^2}\qty(N\expval{wf^2}_p + N(N-1)\expval{f}_p^2) - \expval{f}^2_p
            \\
            &= \frac{1}{N}\qty(\expval{wf^2}_p - \expval{f}_p^2) \label{eq:unbiased_sem_helper2}
        \end{align}
    \end{subequations}
    Now we show that the expected value of \cref{eq:unbiased_sem} is equivalent to $\Var_q(\widehat{f})$.
    \begin{subequations}
        \begin{align}
           \expval{\widehat{\Var}_q(\widehat{f})}_q &= \expval{\frac{1}{N-1}\qty(\frac{1}{N}\qty[\sum_{i=1}^N w_i^2 f(y_i)^2] - (\widehat{f})^2)}_q
           \\
           &= \frac{1}{N-1} \qty(\expval{wf^2}_p - \expval{(\widehat{f})^2}_q)
           \intertext{Using \cref{eq:unbiased_sem_helper1,eq:unbiased_sem_helper2}, we have:}
           &= \frac{1}{N-1} \qty(\qty[N \Var_q(\widehat{f}) + \expval{f}_p^2] - \qty[\Var_q(\widehat{f}) + \expval{\widehat{f}}^2_q])
           \\
           &= \Var_q(\widehat{f})
        \end{align}
    \end{subequations}
\end{proof}

\begin{proposition}\label{prop:var_est}
    \Cref{eq:unbiased_var} is an unbiased estimator of $\Var_p(f)$.
\end{proposition}
\begin{proof}
    Using results from \cref{prop:unbiased_sem}, we have:
    \begin{subequations}
        \begin{align}
            \expval{\widehat{f^2} - (\widehat{f})^2 + \widehat{\Var}_q(\widehat{f})}_q &= \expval{f^2}_p - \qty(\Var_q(\widehat{f}) + \expval{\widehat{f}}^2_q) + \Var_q(\widehat{f})
            \\
            &= \expval{f^2}_p - \expval{f}^2_p
            \\
            &= \Var_p(f)
        \end{align}
    \end{subequations}
\end{proof}
\end{widetext}

\end{document}

%% file: fcc.tex
\tdplotsetmaincoords{70}{30}

\begin{tikzpicture}[tdplot_main_coords, scale=3.2]

    \tikzset{
        site/.style  = {circle, fill=black, inner sep=0pt, minimum size=5pt},
        cub/.style    = {line width=0.75pt, gray!65},             
        cubh/.style   = {line width=0.60pt, gray!38, dashed},     
        k1/.style  = {line width=2.0pt, Blue},
        k2/.style  = {line width=2.0pt, Green},
        k3/.style  = {line width=2.0pt, Red},
        k4/.style  = {line width=2.0pt, Aquamarine},
        k5/.style  = {line width=2.0pt, Orange},
        k6/.style  = {line width=2.0pt, Purple},
    }
    
    \draw[cubh] (0,0,0) -- (0,1,0);
    \draw[cubh] (1,1,0) -- (0,1,0);
    \draw[cubh] (0,1,1) -- (0,1,0);
    
    \draw[cub] (0,0,0) -- (1,0,0) -- (1,1,0);                      
    \draw[cub] (0,0,0) -- (0,0,1);                                 
    \draw[cub] (1,0,0) -- (1,0,1);                                 
    \draw[cub] (1,1,0) -- (1,1,1);                                 
    \draw[cub] (0,0,1) -- (1,0,1) -- (1,1,1) -- (0,1,1) -- cycle;  
    
    %
    %
    \draw[k1] (0,0,0) -- (1,0,1)
        node[pos=0.3, font=\footnotesize, text=Blue, left=2pt] {$K_1$};
    
    \draw[k2] (1,0,0) -- (0,0,1)
        node[pos=0.7, font=\footnotesize, text=Green, left=1pt] {$K_2$};
    
    %
    \draw[k3] (0,0,1) -- (1,1,1)
        node[pos=0.3, font=\footnotesize, text=Red, above left=-2pt] {$K_3$};
    
    \draw[k4] (1,0,1) -- (0,1,1)
        node[pos=0.7, font=\footnotesize, text=Aquamarine, above right=-3pt] {$K_4$};
    
    %
    \draw[k5] (1,0,0) -- (1,1,1)
        node[pos=0.7, font=\footnotesize, text=Orange, right=-2pt] {$K_5$};
    
    \draw[k6] (1,1,0) -- (1,0,1)
        node[pos=0.4, font=\footnotesize, text=Purple, right=-1pt] {$K_6$};
    
    \node[site] at (0,0,0) {};
    \node[site] at (1,0,0) {};
    \node[site] at (1,1,0) {};
    \node[site] at (0,0,1) {};
    \node[site] at (1,0,1) {};
    \node[site] at (0,1,1) {};
    \node[site] at (1,1,1) {};
    
    \node[site] at (0.5,0.5,1) {};   
    \node[site] at (0.5,0,0.5) {};   
    \node[site] at (1,0.5,0.5) {};   
    
\end{tikzpicture}